\documentclass[11pt]{article}

\usepackage[margin=1in]{geometry}

\usepackage{amsmath}
\usepackage{amssymb}
\usepackage{algorithm,algorithmic}
\usepackage{dblfloatfix}
\usepackage{empheq}

\usepackage{hyperref}
\usepackage[capitalize]{cleveref}

\usepackage{thm-restate}

\newtheorem{theorem}{Theorem}
\newtheorem{corollary}[theorem]{Corollary}

\newtheorem{lemma}[theorem]{Lemma}

\newtheorem{definition}[theorem]{Definition}
\newtheorem{proposition}[theorem]{Proposition}

\newtheorem{remark}[theorem]{Remark}

\newenvironment{proof}{\noindent\bf{Proof.}\rm}{\hfill$\blacksquare$\bigskip}

\usepackage[]{color-edits}
\addauthor{UF}{red}

\addauthor{LT}{blue}

\definecolor{darkgreen}{rgb}{0.0, 0.6, 0.03}
\addauthor{AS}{darkgreen}

\newcommand{\ase}[1]{{\ASedit{#1}}}

\newcommand{\OLD}[1]{}

\newcommand{\lemref}[1]{Lemma~\protect\ref{lem:#1}}

\newcommand{\thmref}[1]{Theorem~\protect\ref{thm:#1}}

\begin{document}

\title{A tight negative example for MMS fair allocations}

\author{Uriel Feige\thanks{Weizmann Institute, Israel. \texttt{uriel.feige@weizmann.ac.il}}, Ariel Sapir\thanks{ \texttt{arielsa@cs.bgu.ac.il}}, Laliv Tauber\thanks{Bar Ilan University, Israel. \texttt{lalivta@gmail.com}}}


\maketitle

\begin{abstract}
We consider the problem of allocating indivisible goods to agents with additive valuation functions. Kurokawa, Procaccia and Wang {[JACM, 2018]} present instances for which every allocation gives some agent less than her maximin share. We present such examples with larger gaps. For three agents and nine items, we design an instance in which at least one agent does not get more than a $\frac{39}{40}$ fraction of her maximin share.
{Moreover, we show that there is no negative example in which the difference between the number of items and the number of agents is smaller than six, and that the gap (of $\frac{1}{40}$) of our example is worst possible among all instances with nine items.}

For $n \ge 4$ agents, we show examples in which at least one agent does not get more than a $1 - \frac{1}{n^4}$ fraction of her maximin share. {In the instances designed by Kurokawa, Procaccia and Wang, the gap is exponentially small in $n$.}

{Our proof techniques extend to allocation of chores (items of negative value), though the quantitative bounds for chores are different from those for goods. For three agents and nine chores, we design an instance in which the MMS gap is $\frac{1}{43}$.}
\end{abstract}

\section{Introduction}

We consider allocation problems with $m$ items, $n$ agents, and nonnegative additive valuation functions. The {\em maximin share} (MMS) of an agent $i$ is the highest value $w_i$, such that if all agents have the same valuation function $v_i$ that $i$ has, there is an allocation in which every agent gets value at least $w_i$. An allocation is {\em maximin fair} if every agent gets a bundle that she values at least as much as her MMS. Hence if all agents have the same valuation function, a maximin fair allocation exists.

Perhaps surprisingly, if agents have different additive valuation functions, then a maximin fair allocation need not exist.
Kurokawa, Procaccia and Wang~\cite{KPW18} present negative examples showing that for every $n \ge 3$, there are instances for which in every allocation, at least one agent does not receive her MMS. The {\em gap} (namely, the fraction of MMS lost by some agent) is not stated explicitly in~\cite{KPW18}. However, one can derive explicit gaps from their examples by substituting values for certain parameters that are used in the examples. Doing so gives gaps that are exponentially small in $n$. Even for small $n$, the gaps shown by these examples are orders of magnitude smaller than the positive results that are known, where these positive results show that (for additive valuations) there always is an allocation that gives every agent at least a $\frac{3}{4} + \Omega(\frac{1}{n})$ of her MMS~\cite{KPW18,BK20,GHSSY18,GT20}.

We present negative examples with substantially larger gaps than those shown in~\cite{KPW18}. The motivation for designing such examples is that they are needed if one is to ever establish tight bounds on the fraction of the MMS that can be guaranteed to be given to agents. Though we are far from establishing tight bounds for general instances, our bounds are tight in special cases. {In particular, when there are at most nine items and the additive valuation functions are
integer valued, our results imply the following tight threshold phenomenon. If for every agent, {her valuation function is such that the sum of all item values} is at most~$119$, then a maximin fair allocation always exists. If the sum of item values is~$120$, then there are instances in which a maximin fair allocation does not exist, and then the gap is $\frac{1}{40}$. If the sum of item values is larger than~$120$, the gap cannot be larger then $\frac{1}{40}$.}


\subsection{Our results}

The term {\em negative example} will refer to an allocation instance with additive valuation functions in which there is no allocation that gives every agent her MMS. The term $Gap(n,m)$ refers to the largest possible value $\delta \ge 0$, such that there is an allocation instance with $m$ items and $n$ agents with additive valuations, such that in every allocation there is an agent that gets at most a $1 - \delta$ fraction of her MMS.

In our work, we find negative examples with the smallest possible number of items. The number of items turns out to be nine. Among allocation instances with nine items, we find the allocation instance with the largest gap. Theorem~\ref{thm:largeGap} is based on this allocation instance.

\begin{theorem}
\label{thm:largeGap}
There is an allocation instance with three agents and nine items for which in every allocation, at least one of the agents does not get more than a $\frac{39}{40}$ fraction of her MMS. In other words, $$Gap(n = 3 \; , \; m = 9) \ge \frac{1}{40}.$$
\end{theorem}

The minimality of the number items in Theorem~\ref{thm:largeGap} is implied by Theorem~\ref{thm:n+5}, together with Proposition~\ref{pro:MMSexists} that implies that $n\ge 3$ in every negative example.

\begin{theorem}
\label{thm:n+5}
For every $n$, every allocation instance with $n$ agents and $m \le n + 5$ items has an allocation in which every agent gets her MMS. In other words, $$Gap(n \ge 1 \; , \; m \le n + 5) = 0$$
\end{theorem}

A weaker version of Theorem~\ref{thm:n+5} (with $m \le n + 3$) was previously proved in~\cite{BL15}.

The maximilaty of the gap in Theorem~\ref{thm:largeGap} (when there are nine items) is implied by Theorem~\ref{thm:tightGap}.

\begin{theorem}
\label{thm:tightGap}
Every allocation instance with three agents and nine items has an allocation in which every agent gets at least a $\frac{39}{40}$ fraction of her MMS. In other words, $$Gap(n = 3 \; , \; m = 9) \le \frac{1}{40}$$
\end{theorem}

The proof of Theorem~\ref{thm:tightGap} is based on analysis that reduces the infinite space of possible negative examples into a finite number of classes. For each class, the negative example with largest possible gap within the class can be determined by solving a linear program. For every class we solved the respective linear program using a standard LP solver, and verified that there is no negative example (with 9 items) for which the gap is larger than $\frac{1}{40}$.

Theorem~\ref{thm:largeGap} is concerned with three agents. We also provide negative examples for every number of agents $n \ge 4$. The gaps in these negative examples deteriorate at a rate that is polynomial in $\frac{1}{n}$.

\begin{theorem}
\label{thm:polynomialGap}
For every $n \ge 4$, there is an allocation instance with $n$ agents and at most $3n + 3$ items for which in every allocation, at least one of the agents does not get more than a $1 - \frac{1}{n^4}$ fraction of her MMS. In other words, $$Gap(n \ge 4 \; , \; m \le 3n+3) \ge \frac{1}{n^4}$$
\end{theorem}

Our negative examples that prove Theorem~\ref{thm:polynomialGap} are inspired by, and contain ingredients from, the negative examples presented in~\cite{KPW18}. The new aspect in our constructions is the formulation of Lemma~\ref{lem:split}, the observation that this lemma suffices for the proofs to go through (a related but more demanding property was used in~\cite{KPW18}), and a design, based on modular arithmetic, that satisfies the Lemma.

{The techniques of this paper extend from allocation of goods to allocation of chores (items of negative value, or equivalently, positive dis-utility). We find that results for chores are qualitatively similar to those for goods, though quantitative values of the gaps are different from those values for goods. Likewise, the proof techniques for the case of chores are similar to those shown in this paper for goods, though some of the details change. To simplify the presentation in this paper, all sections of the paper refer only to allocation of goods, except for Section~\ref{sec:chores} that refers only to allocation of chores. Section~\ref{sec:chores} is kept short, and presents only the adaptation of Theorem~\ref{thm:largeGap} to the case of chores.

\begin{theorem}
\label{thm:chores}
There is an allocation instance with three agents and nine chores for which in every allocation, at least one of the agents does not get less than a $\frac{44}{43}$ fraction of her MMS (of dis-utility). In other words, the instance has an MMS gap of $\frac{1}{43}$.
\end{theorem}

We have verified that with eight chores, there always is an allocation giving every agent no more dis-utility than her MMS, and (using a computer assisted proof) that for nine items, $\frac{1}{43}$ is the largest possible gap. However, we omit details of this verification from this manuscript.}

\subsection{Related work}

In this section we review related work that is most relevant to the current paper. In particular, we shall only review papers that concern the maximin share (there are numerous papers considering other fairness notions), and only in the context of nonnegative additive valuation functions (some of the works we cite consider also other classes of valuation functions).

The maximin share was introduced by Budish~\cite{Budish11}.
The fact that there are allocations instances with additive valuations in which no MMS allocation exists was shown in~\cite{KPW18}. That paper presents an instance with three agents and twelve items that has no MMS allocation. The gap in that instance as presented in that paper is around $10^{-6}$, though by optimizing parameters associated with the instance it is possible to reduce the gap to the order of $10^{-3}$. The paper also shows that for every $n \ge 4$ there are instances with $3n+4$ items and no MMS allocation. The gaps in these instances are exponentially small in $n$, and this is inherent in the construction given in that paper.

Work on proving the existence of allocations that give a large fraction of the MMS was initiated in~\cite{KPW18}. The largest fraction currently known is $\frac{3}{4} + \frac{1}{12n}$~\cite{GT20}.

For the case of three agents, it was shown in~\cite{GM19} that there is an allocation that gives every agent at least a $\frac{8}{9}$ fraction of her MMS. Our Theorem~\ref{thm:largeGap} shows that one cannot guarantee more than a $\frac{39}{40}$ fraction in this case. For the case of four agents, it was shown in~\cite{GHSSY18} that there is an allocation that gives every agent at least a $\frac{4}{5}$ fraction of her MMS.

In~\cite{BL15} it was shown that an MMS allocation always exists if $m \le n+3$. We improve the bound to $m \le n + 5$, and show that this is best possible when $n=3$.

In~\cite{AMNS17} it was shown that if all items have values in $\{0,1,2\}$ then an MMS allocation exists. Our negative example in Theorem~\ref{thm:largeGap} uses integer values as high as~26.

\subsection{Preliminaries}

An allocation instance has a set $M = \{e_1, \ldots, e_m\}$ of $m$ items and a set $\{1, \ldots n\}$ of $n$ agents. The term {\em bundle} will always denote a set of items. Every agent $i$ has a valuation function $v_i: 2^M \rightarrow R$ that assigns a value to every possible bundle of items. We assume throughout that valuation functions $v$ are normalized ($v(\emptyset) = 0$) and monotone ($v(S) \le v(T)$ for all $S \subset T \subseteq M$). An $n$-partition of $M$ is a partition of $M$ into $n$ disjoint bundles. $P_n(M)$ denotes the set of all $n$-partitions of $M$. An allocation $A = (A_1, \ldots, A_n)$ is an $n$-partition of $M$, with the interpretation that for every $1 \le i \le n$, agent $i$ receives bundle $A_i$. The utility that agent $i$ derives from this allocation is $v_i(A_i)$.

\begin{definition}
\label{def:MMS}
Consider an allocation instance with a set $M = \{e_1, \ldots, e_m\}$ of $m$ items and a set $\{1, \ldots n\}$ of $n$ agents. Then the maximin share of agent $i$, denoted by $MMS_i$, is the maximum over all $n$-partitions of $M$, of the minimum value under $v_i$ of a bundle in the $n$-partition.
$$MMS_i = \max_{(B_1, \ldots, B_n) \in P_n(M)} \min_j [v_i(B_j)]$$
An $n$-partition that maximizes the above expression will be referred to as an $MMS_i$-partition.
\end{definition}

An allocation that gives every agent at least her MMS is referred to as an MMS allocation.

A valuation function $v$ is additive if $v(S) = \sum_{e \in S} v(e)$.
Though Definition~\ref{def:MMS} applies to arbitrary valuation functions, in this paper we shall only consider additive valuation functions.

By convention, in all remaining parts of the paper, all valuation functions are additive, unless explicitly stated otherwise.

We now review some known propositions concerning the MMS (with additive valuations). For completeness, we also sketch the proofs of these propositions, though we emphasize that all propositions in this section were known and are not original contributions of the current paper.

\begin{proposition}
\label{pro:MMSexists}
Every allocation instance in which either all agents or all agents but one have the same valuation function has an  MMS allocation.
\end{proposition}

\begin{proof}
Let $v = v_1 = \ldots = v_{n-1}$ be the valuation function shared by all agents but one, and let $v_n$ be the valuation function of agent~$n$ who may have a different valuation function. Let $B_1, \ldots, B_n$ be an MMS partition with respect to $v$. For every agent $i$ with $1 \le i \le n-1$, every one of these bundles has value at least $MMS_i$. Allocate to agent~$n$ the bundle $j$ that maximizes $v_n(B_j)$, and allocate the remaining bundles to the other agents. Additivity of $v_n$ implies that $v_n(B_j) \ge MMS_n$, and hence every agent gets at least her MMS.
\end{proof}

The following three propositions concern reduction steps that allow us to replace an allocation instance by a simpler one.

An allocation instance with additive valuations and $m$ items $\{e_1, \ldots, e_j\}$ is {\em ordered} if for every agent $i$ and every two items $e_j$ and $e_k$ with $j < k$ we have that $v_i(e_j) \ge v_i(e_k)$.
Given an unordered allocation instance  with additive valuations and $m$ items, its {\em ordered version} is obtained by replacing the valuation function $v_i$ of each agent $i$ by a new additive valuation function $v'_i$ in which item values are non-increasing. That is, let $\sigma$ denote a permutation over $m$ items with respect to which the values of items are non-increasing under $v_i$. Then for every $1 \le j \le m$ we have that $v'_i(e_j) = v_i(e_{\sigma^{-1}(j)})$. The following proposition is due to~\cite{BL15}.

\begin{proposition}
\label{pro:BL}
For every instance $I$ with additive valuations, every allocation $A'$ for its ordered
version $I'$ can be transformed to an allocation $A$ for $I$, while ensuring that every agent derives at least as high utility from $A$ in $I$ as derived from $A'$
in $I'$.
\end{proposition}

\begin{proof}
A choosing sequence is a sequence of names of agents (repetitions are allowed).
The choosing sequence induces an allocation by the following procedure. Starting from round~1, in each round $r$, the agent whose name
appears in the $r$th location in the choosing sequence receives the item of highest value for the agent (ties can be broken arbitrarily), among the yet unallocated items. The allocation $A'$ for $I'$ induces a choosing sequence, where for every $r$, the agent in location $r$ is the one to which $A'$ allocated the $r$th most valuable item in $I'$.
Using this choosing sequence for the instance $I$, in every round $r$, the respective agent gets an item that she values at least as her $r$th most valuable item, which is the value of the item that she got under $A'$.
\end{proof}

Proposition~\ref{pro:BL} implies that when searching for a negative example with the maximum possible gap, it suffices to restrict attention to ordered instances.

The following two propositions are helpful for arguments that are based on induction on $n$. As each such proposition concerns two instances, in the MMS notation we shall specify which instance we refer to.

\begin{proposition}
\label{pro:reduce1}
Let $I$ be an arbitrary allocation instance with a set $M$ of items and $n$ agents. Let $I'$ be an allocation instance derived from $I$ by removing an arbitrary item $e$ from $M$, and removing one arbitrary agent. Then for each of the remaining agent $q$, $MMS_q(I') \ge MMS_q(I)$.
\end{proposition}

\begin{proof}
Let $(B_1, \ldots, B_n)$ be an $MMS_q(I)$ partition. By renaming bundles, we may assume without loss of generality that $e \in B_n$. Then $(B_1, \ldots, B_{n-2}, (B_{n-1} \cup B_n \setminus \{e\}) )$ is an $(n-1)$ partition for $M \setminus \{e\}$ that certifies that $MMS_q(I') \ge MMS_q(I)$.
\end{proof}

\begin{proposition}
\label{pro:reduce2}
Let $I$ be an arbitrary allocation instance with a set $M$ of $m \ge 2$ items, and $n$ agents. Let $I'$ be an allocation instance derived from $I$ by removing two items $e_i$ and $e_j$ from $M$, and removing one arbitrary agent. Then for every remaining agent $q$, if {either the $MMS_q(I)$ partition has a bundle that contains both $e_i$ and $e_j$, or $v_q(e_i) + v_q(e_j) \le MMS_q(I)$,} then $MMS_q(I') \ge MMS_q(I)$.
\end{proposition}

\begin{proof}
Let $(B_1, \ldots, B_n)$ be an $MMS_q(I)$ partition for $I$. If both $e_i$ and $e_j$ belong to the same bundle, then the proof is as in that for Proposition~\ref{pro:reduce1}. If $e_i$ and $e_j$ are in different bundles, by renaming bundles, we may assume without loss of generality that $e_i \in B_{n-1}$ and $e_j \in B_n$. Then $(B_1, \ldots, B_{n-2}, (B_{n-1} \setminus \{e_i\}) \cup (B_n \setminus \{e_j\}) )$ is an $(n-1)$ partition for $M \setminus \{e_i, e_j\}$ that certifies that $MMS_q(I') \ge MMS_q(I)$. This is because $v_q\left((B_{n-1} \setminus \{e_i\}) \cup (B_n \setminus \{e_j\})\right) = v_q(B_{n-1}) + v_q(B_n) - (v_q(e_i) + v_q(e_j)) \ge 2MMS_q(I) - MMS_q(I) = MSS_q(I)$.
\end{proof}

\section{An MMS gap of $\frac{1}{40}$}
\label{sec:40}


In this section we prove Theorem~\ref{thm:largeGap}, showing an allocation instance for which in every allocation, at least one of the agents gets at most a $\frac{39}{40}$ fraction of her MMS.

\begin{proof}
To present the instance that proves Theorem~\ref{thm:largeGap}, we think of the nine items as arranged in a three by three matrix, with rows $r_1, r_2, r_3$ (starting from the top) and columns $c_1, c_2, c_3$ (starting from the left).

$\left(
  \begin{array}{ccc}
    e_1 & e_2 & e_3 \\
    e_4 & e_5 & e_6 \\
    e_7 & e_8 & e_9 \\
  \end{array}
\right)$

There are three agents, referred to as $R$ (the {\em row} agent), $C$ (the {\em column} agent), and $U$ (the {\em unbalanced} agent).
The MMS of every agent is~40. When depicting valuation functions, for each agent, we present the items in one of her MMS bundles in boldface.

Every row in the valuation function of $R$ has value~40 and gives $R$ her MMS. Her valuation function is:

$\left(
  \begin{array}{ccc}
    1 & 16 & 23 \\
    26 & 4 & 10 \\
    {\bf 12} & {\bf 19} & {\bf 9} \\
  \end{array}
\right)$

Every column in the valuation function of $C$ has value~40 and gives $C$ her MMS. Her valuation function is:

$\left(
  \begin{array}{ccc}
    1 & 16 & {\bf 22} \\
    26 & 4 & {\bf 9} \\
    13 & 20 & {\bf 9} \\
  \end{array}
\right)$

The bundles that give $U$ her MMS are $p = \{e_2, e_4\}$ (the {\em pair}, in boldface), $d = \{e_3, e_5, e_7\}$ (the {\em diagonal}), and $q = \{e_1, e_6, e_8, e_9\}$ (the {\em quadruple}). The valuation function of $U$ is:

$\left(
  \begin{array}{ccc}
    1 & {\bf 15} & 23 \\
    {\bf 25} & 4 & 10 \\
    13 & 20 & 9 \\
  \end{array}
\right)$

It remains to show that no allocation gives every agent her MMS. An allocation is a partition into three bundles. As a sanity check, let us first consider the three partitions that each give one of the agents her MMS. For the partition $(r_1,r_2,r_3)$, both $C$ and $U$ want only $r_3$, and hence one of them does not get her MMS. For the partition $(c_1,c_2,c_3)$, both $R$ and $U$ want only $c_3$, and hence one of them does not get her MMS. For the partition $(p,d,q)$, both $R$ and $C$ want only $p$, and hence one of them does not get her MMS.

To analyse all possible partitions in a systematic way, we consider a valuation function $M$ that values each item as the maximum value given to the item by the three agents. Hence $M$ is:

$\left(
  \begin{array}{ccc}
    1 & 16 & 23 \\
    26 & 4 & 10 \\
    13 & 20 & 9 \\
  \end{array}
\right)$

Every allocation that gives every agent her MMS partitions $M$ into three bundles, where the sum of values in each bundle is at least~40, but not more than~42 (as the sum of all values of $M$ is $40*3 + 2$). If one of the bundles has two items, then this bundle must be $\{e_2, e_4\} = p$, whose value under $M$ is~42. Hence each of the two remaining bundles must have value~40 under $M$.
The unique way of partitioning the remaining items into two bundles of value~40 is to have the bundles $\{e_3, e_5, e_7\} = d$ and $\{e_1, e_6, e_8, e_9\} = q$. {(The only way of reaching a value~40 in a bundle that contains item $e_3$ of value 23 is to include the two items of values~4 and~13.)} But we already saw (in the sanity check) that the partition $(p, d ,q)$ is not a valid solution.

It follows that the partition must be into three bundles, each of size three. The bundle containing $e_9$ must have value between~40 and~42. There are only two such bundles of size three, namely $r_3$ and $c_3$. Each of them has value~42. If one of them is chosen, the remaining two bundles in the partition must then each be of value~40. For $e_4$, the only two bundles of value~40 are $r_2$ and $c_1$. Hence we get only two possible partitions, $(r_1,r_2,r_3)$ and $(c_1,c_2,c_3)$, and both were already excluded in our sanity check.
\end{proof}

\section{MMS gaps that are inverse polynomial in the number of agents}

We present examples that apply for every $n \ge 4$. The initial design of our examples will include $5n - 7$ items, but for $n \ge 6$, this number will be reduced later.  It will be convenient to think of the items as being arranged as selected entries in an $n$ by $n$ matrix, along the perimeter of the matrix, and along its main diagonal. We will construct two valuation functions, where a set $R$ of at least two agents have valuation function $V_R$, and a set $C$ of of at least two agents have valuation function $V_C$ (here $R$ stands for {\em row} and $C$ stands for {\em column}, and $|R| + |C| = n$). We will start with a base matrix $B$,
and then modify $B$ so as to obtain $V_R$ and $V_C$.

In the base matrix $B$, the items have only seven different values, regardless of the value of $n$. We shall partition the items into groups of items of equal value, and give an informative name to each group.

Rows are numbered from top down, and columns from left to right. We use the convention that the index $j$ specifies an arbitrary value in the range $2\le j \le n-1$.

The value of items in each group, and the locations of the groups in $B$, are as follows.

\begin{itemize}

\item $B_{1j} = (n-2)n$. (Top row, excluding corners).

\item $B_{1n} = 1$. (Top-right corner.)

\item $B_{j1} = (n-2)(n-1)$. (Left column, excluding corners.)

\item $B_{jj} = (n-2)(n^2 - 4n + 2)$. (Main diagonal, excluding corners.)

\item $B_{jn} = (n-2)(n-1) + 1$. (Right column, excluding corners.)

\item $B_{n1} \cup B_{nj} = (n-2)^2 + 1$. (Bottom row, excluding bottom-right corner).

\item $B_{nn} = (n-2)(n-3)$. (Bottom-right corner.)

\end{itemize}

For $n=4$, this gives the following matrix.

$\left(
  \begin{array}{cccc}
    0 & 8 & 8 & 1 \\
    6 & 4 & 0 & 7 \\
    6 & 0 & 4 & 7 \\
    5 & 5 & 5 & 2 \\
  \end{array}
\right)$


Observe that all entries of $B$ are nonnegative. Moreover,  All row sums and all column sums have the same value $t_B = n(n-2)^2 + 1$ 

A bundle of items will be called {\em good} if the sum of its values is $t_B$. Hence all rows and all columns are good, but there are also other bundles that are good. A partition of all items into $n$ bundles is {\em good} if every bundle in the partition is good. For example, a partitioning of the items into row bundles is good, and likewise, a partitioning into column bundles is good. The following lemma constrains the structure of good partitions of $B$.

\begin{lemma}
\label{lem:split}
In every partitioning of the items of $B$ into $n$ {\em good} bundles, the structure of the good partition is such that at least one of the following three conditions hold:

\begin{enumerate}
    \item The bottom row is split among the $n$ good bundles (one item in each bundle).
    \item The right column is split among the $n$ good bundles (one item in each bundle).
    \item At least one of the bundles contains at least one item from the bottom row and at least one item from the right column, but does not contain the item $B_{nn}$.
\end{enumerate}
\end{lemma}

\begin{proof}
Observe that $t_B = n(n-2)^2 + 1 = 1$ modulo $n-2$. There are exactly $2n-2$ items that have value~1 modulo $n-2$ (the bottom row and the right column, excluding the bottom-right corner). We refer to these items as {\em special}. The remaining items have value~0 modulo $n-2$, and are not special.
In every good partition, it must be the case that one good bundle has $n-1$ special items, and each other good bundle has one special item.

Consider the good bundle with $n-1$ special items.

If the $n-1$ special items are all in the bottom row (or all in the right column), then item $B_{nn}$ must be the remaining item in the bundle (that is the only way to reach $t_B$), and then the right column (or bottom row) must be split.

If the $n-1$ odd items include at least one from the bottom row and at least one from the right column, then we may assume that $B_{nn}$ is also in the bundle (as otherwise condition~3 of the Lemma holds). This accounts for $n$ items in the bundle. The sum of values of these $n$ items cannot possibly be equal to $t_B$. This can be verified by a case analysis. If $B_{1n}$ is among these items, then the only way to reach $t_B$ with $n-2$ additional special items is to add all items of $B_{jn}$ (as special items in the bottom row have strictly smaller value than items in $B_{jn}$), but then the bundle has no special items from the bottom row. Alternatively, if $B_{1n}$ is not among these items, then the only way to reach $t_B$ with $n-1$ special items is to add all special items of the bottom row (as special items in $B_{jn}$ have strictly larger value than special items items in the bottom row), but then the bundle has no special items from the right column.

Consequently, the sum values of these $n$ items needs to be strictly smaller than $t_B$.
Their total value is minimized if they are $B_{1n} \cup B_{nj} \cup B_{nn}$, giving a value of $1 + (n-2)((n-2)^2 + 1) + (n-2)(n-3) = (n-1)(n-2)^2 + 1$. Hence a value of $(n-2)^2$ is missing in order to complete the sum of values to $t_B = n(n-2)^2 + 1$. For $n \ge 5$, none of the remaining items has such small value, and hence such a good bundle cannot be formed at all. The only case that remains to be considered is $n=4$, because for $n=4$ the value of diagonal items $B_{jj}$ happens to satisfy $(n-2)(n^2 - 4n + 2) = (n-2)^2$.

Recall the matrix for $n=4$ depicted above. The composition of values in a good bundle that has two special items from the bottom row, the special item $B_{14}$, the item $B_{44}$, and one diagonal item, is  $(5,5,1,2,4)$.  But then one of the two items of value~8 does not have a good bundle. (An item of value~8 needs an additional value of~9 to reach~17. However, of the items that remain, there is only one combination of items that gives value~9, namely, as $4+5$.)
\end{proof}

\begin{remark}
\label{rem:KPW18}
Our proof for Theorem~\ref{thm:polynomialGap} follows a pattern used in~\cite{KPW18}. In their construction, the base matrix $B$ was required to have the property that it has only two good partitions: the row partition and the column partition. In contrast, we allow $B$ to have many more good partitions (as specified in Lemma~\ref{lem:split}), and show that even with this extra flexibility, the proof pattern of~\cite{KPW18} still works. Given this extra flexibility in the properties of $B$, we design such  matrices (one for each value of $n$) with much smaller integer entries than the corresponding matrices designed in~\cite{KPW18}.
\end{remark}

Using the matrix $B$, we shall now create two matrices, one for $V_R$ and one for $V_C$.
First, every entry of $B$ is multiplied by~$n$. Then, for $V_R$, subtract $1$ from the value of each special item in the bottom row, and add $n - 1$ to the value of the bottom-right corner. For $n=4$, the matrix for $V_R$ is:

$\left(
  \begin{array}{cccc}
    0 & 32 & 32 & 4 \\
    24 & 16 & 0 & 28 \\
    24 & 0 & 16 & 28 \\
    19 & 19 & 19 & 11 \\
  \end{array}
\right)$

The maximin share of every agent in $R$ is~68 (each row is a bundle). For general $n \ge 4$, this maximin share is $t_V = nt_B = n^2(n-2)^2 + n$.

For $V_C$, subtract $1$ from the value of each special item in the right column, and add $n - 1$ to the value of the bottom-right corner. For $n=4$, the matrix for $V_C$ is:

$\left(
  \begin{array}{cccc}
    0 & 32 & 32 & 3 \\
    24 & 16 & 0 & 27 \\
    24 & 0 & 16 & 27 \\
    20 & 20 & 20 & 11 \\
  \end{array}
\right)$

Similar to agents in $R$, the maximin share of every agent in $C$ is~68 (each column is a bundle). For general $n \ge 4$, this maximin share is $t_V = nt_B  = n^2(n-2)^2 + n$.

\begin{proposition}
\label{pro:polynomialGap}
If $|R| + |C| = n$ and $|R|, |C| \ge 2$, then in every allocation, at least one player gets a bundle that he values as at most $t_V - 1$.
\end{proposition}

\begin{proof}
The allocation partitions the items into $n$ bundles. If at least one of the bundles has value less than $t_B$ in $B$, then the same bundle has value at most $n(t_B - 1) + (n - 1) = t_V - 1$ for the agent who receives it. Hence we may assume that every bundle has value $t_B$ in $B$. By Lemma~\ref{lem:split}, there are only three possibilities for this.

\begin{enumerate}
    \item {\em The bottom row is split.} Then every agent in $R$ receives a bundle that contains a single item from the bottom row. As $|R| \ge 2$, for at least one row agent, this single item lost a value of~1 in the process of constructing $V_R$. Consequently, the value received by this agent is $nt_B -  1 = t_V - 1$.

    \item {\em The right column is split.} Then every agent in $C$ receives a bundle that contains a single item from the right column. As $|C| \ge 2$, for at least one column agent, this single item lost a value of~1 in the process of constructing $V_C$. Consequently, the value received by this agent is $nt_B - 1 = t_V - 1$.

    \item {\em At least one of the bundles contains at least one item from the bottom row and at least one item from the right column, but does not contain the item $B_{nn}$.} Such a bundle has value at most $t_V-1$ for every agent.
\end{enumerate}

\end{proof}

We can now prove Theorem~\ref{thm:polynomialGap}. In fact, we state a somewhat stronger version of it in which the gap is improved from $\frac{1}{n^4}$ to a somewhat larger value.

\begin{theorem}
\label{thm:polynomialGap1}
For given $N$, let $n = \lceil \frac{N + 4}{2} \rceil$. Then for every $N \ge 4$, there is an allocation instance with $N$ agents and at most $N + 4n - 7$ items (which gives $3N+1$ when $N$ is even and $3N + 3$ when $N$ is odd) for which in every allocation, at least one of the agents does not get more than a $1 - \frac{1}{f(n)}$ fraction of her MMS. Here, the function $f(n)$ has value $f(n) = n^2(n-2)^2 +n$. In other words, $$Gap(N \ge 5 \; , \; m \le 3N+3) \ge \frac{1}{f(n)}$$
\end{theorem}

\begin{proof}
For $4 \le N \le 5$ we have that the corresponding value of $n = \lceil \frac{N + 4}{2} \rceil = N$, and hence the corresponding instances was described above. (Observe that $f(n)$ equals the corresponding value of $t_V$ in these instances.)

For $N \ge 6$, we have that $n = \lceil \frac{N + 4}{2} \rceil < N$. In this case we construct an instance as above  for the corresponding value of $n$ (with value $t_V = n^2(n-2)^2 +n$). We add to this instance $N-n$ agents so that the number of agents becomes $N$. Among the agents, we set $\lfloor \frac{N}{2} \rfloor$ agents to be row agents, and the remaining agents to be column agents. We also add to the instance $N-n$ auxiliary items, each of value $t_V$, and so the total number of items is $(5n - 7) + (N-n) = N + 4n - 7$.

For each of the $N$ agents, the MMS is $t_v$ (by partitioning the set of items into the $N-n$ auxiliary items, and either the $n$ rows or the $n$ columns). $N-n$ agents get their MMS by getting an auxiliary item. However, among the $n$ agents that remain, at least two are row agents (because $|R| - (N - n) = \lfloor \frac{N}{2} \rfloor - N + \lceil \frac{N + 4}{2} \rceil = 2$) and at least two are column agents, and this suffices for Proposition~\ref{pro:polynomialGap} to apply.
\end{proof}


\section{An MMS allocation whenever $m \le n + 5$}

In this section we prove Theorem~\ref{thm:n+5}, that if $m \le n+5$ there always is an MMS allocation. The proof makes use of the following two lemmas.

\begin{lemma}
\label{lem:reduce1}
Let $I$ be an allocation instance with $n$ agents and $m$ items, and assume that for every instance with $n-1$ agents and $m-1$ items there is an MMS allocation. If there is an agent $i$ and item $e$ for which $v_i(e) \ge MMS_i(I)$, then $I$ has an MMS allocation.
\end{lemma}

\begin{proof}
Remove item $e$ and agent $i$, resulting in an instance $I'$ with $n-1$ agents and $m-1$ items. By Proposition~\ref{pro:reduce1}, for every agent $j \not= i$ it holds that $MSS_j(I') \ge MSS_j(I)$. By the assumption of the lemma, there is an MSS allocation $A'$ for $I'$. Extend $A'$ to an allocation $A$ for $I$, by giving item $e$ to agent $i$. Allocation $A$ is an MSS allocation for $I$.
\end{proof}

\begin{lemma}
\label{lem:reduce2}
Let $I$ be an allocation instance with $n$ agents and $m$ items, and assume that for every instance with $n-1$ agents and $m-2$ items there is an MMS allocation. Suppose that there is an agent $q$ and a bundle $B$ containing two items such that $v_q(B) \ge MMS_q(I)$, and moreover, for every agent $j \not= q$, at least one of the following conditions hold:
\begin{enumerate}
    \item $B$ is {\em small}: $v_j(B) \le MSS_j(I)$.
    \item $B$ is {\em directly dominated}: $B$ is equal to or contained in one of the bundles of the $MMS_j$ partition.
    \item $B$ is {\em indirectly dominated}: the $MMS_j$ partition contains a bundle $B'$ such that $v_j(B') \ge v_j(B)$ and $|B' \cap B| = 1$.
\end{enumerate}
Then $I$ has an MMS allocation.
\end{lemma}

\begin{proof}
Remove bundle $B$ and agent $q$, resulting in an instance $I'$ with $n-1$ agents and $m-2$ items. We claim that $MSS_j(I') \ge MSS_j(I)$ for every agent $j \not= q$. For agents for which either condition~1 or condition~2 hold, this follows by Proposition~\ref{pro:reduce2}.

For an agent $j$ for which only condition~3 holds, let $B_1'$ denote the other bundle intersected by $B$. Replace the two bundles $B'$ and $B_1'$ in the $MMS_j$ partition by the two bundles $B$ and $B_1 = (B' \cup B'_1) \setminus B$. We have that $v_j(B)\ge MMS_j(I)$ (as condition~1 is assumed not to hold) and $v_j(B_1) \ge MMS_j(I)$. {(The last inequality can be verified as follows. Condition~3 holding implies that $v_j(B') \ge v_j(B)$. This together with $(B \cup B_1) = (B' \cup B'_1)$ implies that $v_j(B_1) \ge v_j(B'_1)$. The fact that $B'_1$ is a bundle in the original $MSS_j$ partition implies that $v_j(B'_1) \ge MMS_j$.)}  Hence we get an $MMS_j$ partition in which $B$ is one of the bundles, and now we can apply condition~2 to conclude that $MSS_j(I') \ge MSS_j(I)$.

By the assumption of the lemma, there is an MSS allocation $A'$ for $I'$. Extend $A'$ to an allocation $A$ for $I$, by giving bundle $B$ to agent $q$. Allocation $A$ is an MSS allocation for $I$.
\end{proof}

We now prove Theorem~\ref{thm:n+5}.

\begin{proof}
The proof is by induction on $n$. The theorem trivially holds for $n = 1$, and holds for $n = 2$ by Proposition~\ref{pro:MMSexists}. The case $n=2$ serves as the base case of the induction, and it remains to prove the theorem for $n\ge 3$. In all cases with $n\ge 3$ we assume without loss of generality:

\begin{itemize}
    \item The theorem has already been proved for all $n' < n$ (the inductive hypothesis).
    \item $m = n+5$ (because if $m<n+5$, we may add $n + 5 - m$ auxiliary items that have~0 value to all agents).
    \item All bundles in the MMS partition of every agent are of size at least~2.
\end{itemize}

The third assumption can be made without loss of generality, as otherwise there is an agent $i$ and item $e$ for which $v_i(e) \ge MMS_i$, and then Lemma~\ref{lem:reduce1} allows us to reduce the instance to one in which the induction hypothesis already holds.

{Observe that the third assumption implies (among other things) that it suffices to consider only $n \le 5$, because for $n \ge 6$ we have that $m = n + 5 < 2n$, and the third assumption cannot hold.}

{Using these assumptions, the cases $n = 3$, $n=4$, and $n=5$ are proved in Lemma~\ref{lem:3}, Lemma~\ref{lem:4}, and Lemma~\ref{lem:5}, respectively.}
\end{proof}

\subsection{Three agents, eight items}

\begin{lemma}
\label{lem:3}
Every allocation instance with $n=3$ agents and $m = n + 5$ items has an MMS allocation.
\end{lemma}

\begin{proof}
By Proposition~\ref{pro:BL} we may assume that the instance is ordered (for every $1 \le i < j \le m$ and every agent $q$,  $v_q(e_i) \ge v_q(e_j)$).

Recall (see the proof of Theorem~\ref{thm:n+5}) that we may assume that the MMS partition of an agent contains only bundles of size at least~2. Consequently, for every agent $j$, her $MMS_j$ partition contains at least one bundle (call it $B_j$) that has exactly two items.

If the three bundles $B_1$, $B_2$ and $B_3$ are disjoint, give each agent her respective bundle, and allocate the two remaining items arbitrarily.

It remains to consider the case that at least two of these bundles intersect. W.l.o.g., let these bundles be $B_1$ and $B_2$.

Suppose that $|B_1 \cap B_2| = 1$. Then as the instance is ordered, all agents agree that one of the two bundles, $B_1$ or $B_2$, is not more valuable than the other. W.l.o.g., let this bundle be $B_1$. Likewise, if $B_1 = B_2$, then also in this case $B_1$ is not more valuable than $B_2$.

There are two cases to consider:

\begin{itemize}
    \item $v_3(B_1) \ge MMS_3$. In this case Lemma~\ref{lem:reduce2} applies with agent~3 serving as agent $q$, and with $B_1$ serving as $B$. Hence an MMS allocation exists.
    \item $v_3(B_1) < MMS_3$. In this case Lemma~\ref{lem:reduce2} applies with agent~1 serving as agent $q$, and with $B_1$ serving as $B$. Hence an MMS allocation exists.
\end{itemize}
\end{proof}

\subsection{Four agents, nine items}

\begin{lemma}
\label{lem:4}
Every allocation instance with $n=4$ agents and $m = n + 5$ items has an MMS allocation.
\end{lemma}

\begin{proof}
Consider an allocation instance $I$ with four agents and a set $M$ of at most nine items.
Recall (see the proof of Theorem~\ref{thm:n+5}) that we may assume that the MMS partition of an agent contains only bundles of size at least~2. Consequently, for every agent $i$, her $MMS_i$ partition contains three bundles of size two, and one bundle of size three.

Let $(B_{1,1}, B_{1,2}, B_{1,3}, B_{1,4})$ denote the $MMS_1$ partition of agent~1, with $|B_{1,1}| = |B_{1,2}| = |B_{1,3}| = 2$ and $|B_{1,4}| = 3$. Suppose that for some $k \le 3$, there is exactly one agent $i \ge 2$ for which $v_i(B_{1,k}) \ge MMS_i$. Then Lemma~\ref{lem:reduce2} applies with agent $i$ serving as agent $q$, and $B_{1,k}$ serving as bundle $B$. Hence an MMS allocation exists.

Likewise, if for some $k \le 3$ there is no agent $i \ge 2$ for which $v_i(B_{1,k}) \ge MMS_i$, Lemma~\ref{lem:reduce2} applies with agent $1$ serving as agent $q$, and $B_{1,k}$ serving as bundle $B$. Hence also in this case an MMS allocation exists.

It follows that we can assume that for each of the bundles $\{B_{1,1}, B_{1,2}, B_{1,3}\}$ there is at most one agent $2 \le i \le 4$ that values it less than her MSS.

Consider now a bipartite graph $G$. Its left hand side contains four vertices, corresponding to the four agents $\{1,2,3,4\}$. Its right hand side has four vertices, corresponding to the four bundles  $\{B_{1,1}, B_{1,2}, B_{1,3}, B_{1,4}\}$. For every $1 \le i,j \le 4$ there is an edge between agent $i$ and bundle $B_{1,j}$ if $v_i(B_{1,j}) \ge MMS_i$. Observe that a perfect matching in $G$ induces an MMS allocation, giving every agent her matched bundle. Hence it suffices to show that $G$ has a perfect matching.

Each of the right hand side vertices $B_{1,k}$ for $1 \le k \le 3$ has degree at least~3 (as at most one agent values it less than her MMS), and $B_{1,4}$ has degree at least~1 (as agent~1 values it at least as $MSS_1$). Hence for every $k \le 3$, every set of $k$ right hand side vertices has at least $k$ left hand side neighbors. Moreover, the set of all right hand side vertices has four left hand side neighbors, as for every agent $i$, at least one of the four bundles has value at least $\frac{1}{4}v_i(M) \ge MMS_i$.  Hence by Hall's condition, $G$ has a perfect matching.
\end{proof}

\subsection{Five agents, ten items}

\begin{lemma}
\label{lem:5}
Every allocation instance with $n=5$ agents and $m = n + 5$ items has an MMS allocation.
\end{lemma}

\begin{proof}
Let $I$ be an arbitrary allocation instance with $5$ agents and $10$ items.
Recall (see the proof of Theorem~\ref{thm:n+5}) that we may assume that the MMS partition of an agent contains only bundles of size at least~2. As $m = 2n$, this implies that for every agent $i$, all bundles of her $MMS_i$ partition are of size two.

By Proposition~\ref{pro:BL} we may assume that the instance is ordered (for every $1 \le i < j \le m$ and every agent $q$,  $v_q(e_i) \ge v_q(e_j)$). For every agent $i$, consider the bundle $B_i$ in her MMS partition that contains the item $e_1$. This gives five bundles (not necessarily all distinct). Among these bundles, consider the bundle $B$ in which the second item of the bundle has highest index (lowest value). Then for every agent $i$ we have that $v_i(B) \le v_i(B_i)$, because the instance is ordered.
Let $q$ be an agent that has $B$ as a bundle in her MMS partition (if there is more that one such agent, pick one arbitrarily).
Lemma~\ref{lem:reduce2} (condition~3 in the lemma) implies $I$ has an MMS allocation.
\end{proof}



\section{Tightness of MMS ratio for nine items}

In this section we prove Theorem~\ref{thm:tightGap}, showing that every allocation instance with three agents and nine items has an allocation that gives each agent at least a $\frac{39}{40}$ of her MMS. The proof has three steps.

\begin{enumerate}
\item The proof of Theorem~\ref{thm:structure} that shows that a negative example can have only one of two possible structures. 

\item Each structure induces linear constraints on the valuation functions of the agents. For each structure, we set up a linear program that finds a solution that satisfies all linear constraints implied by the corresponding structure, while maximizing the MMS gap in that solution.
These LPs are under-constrained, and the optimal feasible solutions of these LPs turn out not to correspond to true negative examples. Hence we need to add additional constraints to the LPs, preventing the LPs from producing solutions that are not true negative examples.

\item For each of the two structures, we partition all potential negative examples that have this structure into a finite number of classes, where each class offers some refinement of the structure. The classes need not be disjoint. The refined structure of a class gives rise to additional constraints to the LP. Thus we end up with a finite number of different LPs, one for each class. We then verify that none of these LPs generates a negative example with MMS gap larger than $\frac{1}{40}$ (this is done by having a computer program solve the corresponding LPs), and this proves Theorem~\ref{thm:tightGap}.
\end{enumerate}

\subsection{Only two possible structures}

We first present the theoretical analysis that leads to Theorem~\ref{thm:structure}. 

Let $I$ be an instance with three agents $\{1,2,3\}$ and a set $M$ of nine items $\{e_1, \ldots, e_9\}$.
By Proposition~\ref{pro:BL}, we assume without loss of generality that the instance is {\em ordered} (for every $i < j$, all agents agree that item $e_i$ has value at least as large as item $e_j$).




We say that a bundle $B$ of items is {\em good} for agent $i$ if $v_i(B) \ge MMS_i$, and {\em bad} for $i$ otherwise.

Every agent $i$ has a partition of the items into three bundles $B_{i1}, B_{i2}, B_{i3}$, such that each of these bundles is good for $i$. Fix for each agent such a partition. We may assume that every bundle is of size at least two (recall Lemma~\ref{lem:reduce1}).

\begin{proposition}
\label{pro:containment}
If a bundle in one partition contains a bundle from another partition (containment need not be strict), then $I$ has an MMS allocation.
\end{proposition}

\begin{proof}
W.l.o.g., assume that $B_{1,1} \subseteq B_{2,1}$. If $B_{1,1}$ is good for agent 3, then give $B_{1,1}$ to agent~3. As $(B_{2,2} \cup B_{2,3}) \subseteq (B_{1,2} \cup B_{1,3})$, we have that $v_2(B_{1,2}) + v_2(B_{1,3}) \ge v_2(B_{2,2}) + v_2(B_{2,3}) \ge 2MMS_2$.
Give agent~2 whichever of the bundles $B_{1,2}$ or $B_{1,3}$ she values as at least $MMS_2$, and give the remaining bundle to agent~1, resulting in an MMS allocation.

If $B_{1,1}$ is bad for agent~3, give $B_{1,1}$ to agent~1, and give agent~3 whichever items are in $B_{2,1} \setminus B_{1,1}$, and in addition, whichever bundle she prefers over $B_{2,2}$ and $B_{2,3}$, thus giving her value at least $\frac{1}{2}v_3(M \setminus B_{1,1}) \ge \frac{1}{3}v_3(M) \ge MMS_3$. Give the remaining bundle to Agent~2. Every agent gets at least her MMS.
\end{proof}

We can conclude that the nine bundles are distinct.

\begin{proposition}
\label{pro:symmetricDifference}
If there are two bundles (necessarily, of different agents) $B$ and $B'$ such that $|B \setminus B'| = |B' \setminus B| = 1$, then $I$ has an MMS allocation.
\end{proposition}

\begin{proof}
W.l.o.g., assume that $|B_{1,1} \setminus B_{2,1}| = |B_{2,1} \setminus B_{1,1}| = 1$. As the instance is ordered and the two bundles share all but one item, it follows that all {agents} view one bundle, without loss of generality let it be $B_{2,1}$, as at least as valuable as the other. Observe that if $B_{1,1}$ is given to either agent~1 or agent~3, then agent~2 can partition the remaining items into two bundles that each has value {of at least} $MSS_2(I)$  (in her bundle that contains the item $B_{1,1} \setminus B_{2,1}$, replace that item by the item $B_{2,1} \setminus B_{1,1}$).  If $B_{1,1}$ is good for agent~3, then give $B_{1,1}$ to agent~3. In the instance that remains, each of the remaining agents has two disjoint bundles of value at least as high as her original MMS, and hence an MMS allocation exists. If $B_{1,1}$ is bad for agent~3, give $B_{1,1}$ to agent~1. Agent~2 can then partition the remaining items into two bundles as explained above. Agent~3 chooses among them the bundle that she prefers, thus getting at least $\frac{v_3(M)-v_3(B_{1,1})}{2} \ge MMS_3$. Agent~2 gets the remaining bundle, and hence at least her MMS.
\end{proof}

\begin{corollary}
\label{cor:size3}
Suppose that for two agents $i$ and $j$ every bundle in their MMS partition is of size three. Then either for every bundle of $i$ and every bundle of $j$ the intersection has exactly one item, or $I$ has an MMS allocation.
\end{corollary}

\begin{proof}
If two bundles of different agents are identical, the corollary follows from Proposition~\ref{pro:containment}. If they have two items in their intersection, the corollary follows from Proposition~\ref{pro:symmetricDifference}. If they are disjoint, then another bundle of the second agent is either identical to the bundle of the first agent, or their intersection has two items.
\end{proof}

\begin{proposition}
\label{pro:size3}
If every bundle is of size three, then $I$ has an MMS allocation.
\end{proposition}

\begin{proof}
Let $e$ be the item of smallest value according to the order $\pi$, and for every agent $i$, let $B_i$ be her bundle that contains $e$. Except for containing $e$, these three bundles are disjoint. Hence there are two items, $e_1$ and $e_2$, that are in none of these bundles, and each of them is worth at least as much as $e$ for every agent. Give $B_3$ to agent~3, give $(B_2 \setminus \{e\}) \cup \{e_2\}$ to agent~2, and  give $(B_1 \setminus \{e\}) \cup \{e_1\}$ to agent~1.
\end{proof}

\begin{proposition}
\label{pro:uniqueGood}
Suppose that the following condition does not hold: among the three bundles of agent~1, there is a bundle that is good both for agent 2 and for agent 3, and the remaining two bundles are bad both for agent 2 and agent 3. Then $I$ has an MMS allocation.
\end{proposition}

\begin{proof}
Observe that for every $j \in \{2,3\}$, there is at least one bundle $B$ of agent~1 with $v_j(B) \ge \frac{1}{3}v_j(M) \ge MMS_j$. If the condition of the proposition does hold, then without loss of generality, $B_{1,2}$ is good for agent~2, and $B_{1,3}$ is good for agent~3. For every $j \in \{1,2,3\}$, give bundle $B_{1,j}$ to agent $j$.
\end{proof}

By symmetry, Proposition~\ref{pro:uniqueGood} applies also if we permute the names of agents. Permuting also the names of the bundles, we can thus assume the following:

\begin{itemize}
    \item For every agent $i$, bundle $B_{i,1}$ is good for both other agents, and bundles $B_{i,2}$ and $B_{i,3}$ are bad for both other agents.
\end{itemize}

Observe that for every $i \not= j$, agent $i$ values bundle $B_{j,1}$ strictly more than $B_{i,1}$, because she values $B_{i,2}$ and $B_{i,3}$ strictly more than $B_{j,2}$ and $B_{j,3}$. Likewise, agent $j$ values bundle $B_{i,1}$ strictly more than $B_{j,1}$.

\begin{proposition}
\label{pro:intersectGodBad}
If the good bundle in the partition of agent $i$ does not intersect a bad bundle in the partition of agent $j \not=i$, then $I$ has an MMS allocation.
\end{proposition}

\begin{proof}
Suppose that $B_{1,1}$ does not intersect $B_{2,2}$. Give $B_{1,1}$ to agent~3 (recall that $B_{1,1}$ is good for all agents), give $B_{2,2}$ to agent~2, and give the remaining items to agent~1. These items form a good bundle for agent~1, because we removed from the grand bundle one of his original bundles ($B_{1,1}$) and a bundle ($B_{2,2}$) that is bad for agent~1.
\end{proof}

\begin{proposition}
\label{pro:2good}
If an agent~$i$ has a bundle $B$ of size two in her partition, then either this bundle is good for the other two agents, or $I$ has an MMS allocation.
\end{proposition}

\begin{proof}
{Suppose that for agent~$i$ bundle $B_{i,2}$ (which is bad for the other two agents) is of size two. Then Lemma~\ref{lem:reduce2} implies that an MMS allocation exists.}
\end{proof}

Observe that the combination of Proposition~\ref{pro:uniqueGood} and Proposition~\ref{pro:2good} implies that an agent cannot have two bundles of size two in his partition.

\begin{corollary}
\label{cor:2disjoint}
If $|B_{1,1} = 2|$, then either the good bundles of the other agent are disjoint from $B_{1,1}$, or $I$ has an MMS allocation.
\end{corollary}

\begin{proof}
Suppose that $|B_{1,1} = 2|$ and that $B_{2,1}$ intersects $B_{1,1}$. Then w.l.o.g., $B_{2,2}$ does not intersect $B_{1,1}$, and Proposition~\ref{pro:intersectGodBad} applies.
\end{proof}

\begin{proposition}
\label{pro:22}
If $|B_{1,1}| = |B_{2,1}| = 2$, then $I$ has an MMS allocation.
\end{proposition}

\begin{proof}
By Corollary~\ref{cor:2disjoint},  $B_{1,1}$ and $B_{2,1}$ and $B_{3,1}$ can be assumed to be disjoint from each other. Hence we can give each agent $i$ the bundle $B_{i,1}$.
\end{proof}

Proposition~\ref{pro:22} implies among other things that every instance with three agents and up to eight items has an MMS allocation, a fact that was already proved in Lemma~\ref{lem:3}.

We are now ready to state Theorem~\ref{thm:structure}. However, before doing it, we introduce some notation and terminology. Recall that we may assume that instance $I$ is ordered. We still do so, but we no longer assume that the order is from item $e_1$ to item $e_9$.
Instead the order is according to some permutation $\pi$ that is left unspecified at this point. Instead, our naming convention for items is based on arranging the items in a three by three matrix, and naming the items according to their location in the matrix, as specified below.

$\left(
  \begin{array}{ccc}
    e_1 & e_2 & e_3 \\
    e_4 & e_5 & e_6 \\
    e_7 & e_8 & e_9 \\
  \end{array}
\right)$

The rows of the matrix are referred to as $R_1, R_2, R_3$, starting from the top row, and the columns are referred to as $C_1, C_2, C_3$, starting from the left column. The two main diagonals of the matrix {are} $\{e_1, e_5, e_9\}$ and $\{e_3, e_5, e_7\}$.

\begin{theorem}
\label{thm:structure}
A negative example for three players and nine items
must have the following structure (after appropriately renaming the items). For one agent $R$, the MMS partition is into the three rows ($R_1$, $R_2$ and $R_3$), for one agent $C$, the MMS partition is into the three columns ($C_1$, $C_2$ and $C_3$), and for one agent $U$, the MMS partition is to a bundle $P = \{e_2, e_4\}$ ($P$ stands for {\em pair}), a bundle $D$ that is one of the two main diagonals ($D$ stands for {\em diagonal}), and a bundle $Q$ with the remaining four items ($Q$ stands for {\em quadruple}). Bundle $P$ is good for all agents, whereas $D$ and $Q$ are bad for agents $R$ and $C$. The row and the column that do not intersect $P$ are good for all agents (these are $R_3$ and $C_3$), whereas the remaining rows and columns are good only for the agents that have them in their partition, and bad for the other agents.
\end{theorem}

As there are two main diagonals, Theorem~\ref{thm:structure} offers two possible structures. We refer to them as the parallel diagonals structure (bundle $d$ runs in parallel to bundle $p$), and the crossing diagonals structure (bundle $d$ crosses bundle $p$). They are depicted in figures~\ref{fig:9Elements3PartitionsB} and~\ref{fig:9Elements3Partitions}, respectively. Within each figure, for every item $e_i$, the entry $r_i$ ($c_i$, $u_i$, respectively) denotes its value to agent $R$ ($C$, $U$, respectively).

\begin{figure*}[!h]
  \centering
    \includegraphics[width=4cm, scale=0.5]{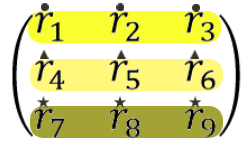}
    \includegraphics[width=4cm, scale=0.5]{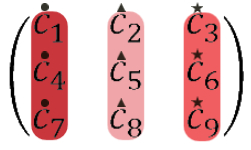}
    \includegraphics[width=4cm, scale=0.5]{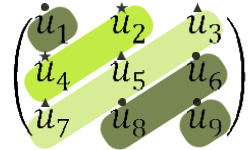}
  \caption{The parallel diagonals structure, with MMS partitions for players $R$, $C$ and $U$ respectively. The good bundles, marked by a $\star $ above the item, are $R_3 = \{e_7, e_8, e_9\}$, $C_3 = \{e_3, e_6, e_9\}$ and $P = \{e_2, e_4\}$.}
  \label{fig:9Elements3PartitionsB}
\end{figure*}

\begin{figure*}[!h]
  \centering
    \includegraphics[width=4cm, scale=0.5]{9Elements_ByRow_ver2.jpg}
    \includegraphics[width=4cm, scale=0.5]{9Elements_ByCol_ver2.jpg}
    \includegraphics[width=4cm, scale=0.5]{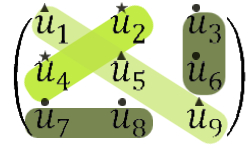}
  \caption{The crossing diagonals structure, with MMS partitions for players $R$, $C$ and $U$ respectively. The good bundles, marked by a $\star $ above the item, are $R_3 = \{e_7, e_8, e_9\}$, $C_3 = \{e_3, e_6, e_9\}$ and $P = \{e_2, e_4\}$. }
  \label{fig:9Elements3Partitions}
\end{figure*}


We now prove Theorem~\ref{thm:structure}.

\begin{proof}
By Proposition~\ref{pro:size3}, for at least one agent her MMS partition contains a bundle of size two. Propositions~\ref{pro:2good} and~\ref{pro:22} imply that at most one  agent has a bundle of size two in her MMS partition. It  follows that there are exactly two agents ($R$ and $C$) for which the MMS partition is composed of bundles of size three, and one agent ($U$) whose partition has a bundle of size two. That partition cannot have two bundles of size two (see remark after Proposition~\ref{pro:2good}), and hence the other two bundles in that partition are of sizes three and four.

By Corollary~\ref{cor:size3}, the intersection of every bundle of $R$ and every bundle of $C$ is of size~1. Hence up to permuting the names of the items, the MMS partition of $R$ can be assumed to be the row bundles, and the MMS partition of $C$ can be assumed to be the column bundles. Each of {the agents} $R$ and $C$ has exactly one bundle that is good for all agents, and without loss of generality these bundles are $R_3$ and $C_3$.

Agent $U$ has a bundle of size two. By Proposition~\ref{pro:2good}, this bundle is good for the other two agents. By Proposition~\ref{pro:intersectGodBad}, it intersects only bad bundles of other agents. By Proposition~\ref{pro:containment} it is not contained in a bundle of any other agent. Hence without loss of generality, this bundle is $\{e_2,e_4\}$, which we denoted by $P$. (The other alternative would have been items $\{e_1, e_5\}$, but it can be converted to $P$ by switching the order among the first two rows.)

As noted above, agent $U$ has a bundle (call it $D$) with three items. It follows from Propositions~\ref{pro:containment} and~\ref{pro:symmetricDifference} that $D$ intersects every bundle of every other agent exactly once. This gives three options. Two of them are the main diagonals, giving the parallel diagonals and the crossing diagonals structures allowed by Theorem~\ref{thm:structure}. The remaining option is $\{e_1, e_6, e_8\}$. However, by renaming the items (permuting rows $R_1$ and $R_2$ and permuting columns $C_1$ and $C_2$)
we see that this is the same structure as the parallel diagonals one.
\end{proof}

\subsection{The LP approach}

Theorem~\ref{thm:structure} establishes that a negative example with three agents and nine items, if one exists, must have one of two possible structures, either the one that we referred to as {\em parallel diagonals}, or the one that we referred to as {\em crossing diagonals}. Indeed, the negative example presented in the proof of Theorem~\ref{thm:largeGap} has the parallel diagonals structure. In this section we explain how each of the structures can guide us in a systematic computer assisted search for an actual negative example.

{Let us first introduce terminology that will be used in our proof. Among all negative examples, we wish to find the one with largest MMS gap. We refer to such a negative example as a {\em max-gap instance}. Equivalently, a max-gap instance is one of smallest value $b$, such that there is an allocation instance in which no allocation gives every agent more than a {$\frac{b-1}{b}$} fraction of her MMS. By scaling valuation functions of agents, this is equivalent to the following definition.

\begin{definition}
\label{def:maxGap}
A {\em max gap instance} is an allocation instance with the smallest value of $b$, in which the MMS of every agent is $b$, but in every allocation, at least one agent gets a value of at most $b-1$. A {\em max-gap PD-instance} ({\em max-gap CD-instance}, respectively) is a max gap instance among those instances that have the parallel diagonals structure (crossing diagonals, respectively).
\end{definition}

By Theorem~\ref{thm:structure}, a max gap instance is either a max-gap PD-instance or a max-gap CD-instance. We will explain in this section how to bound $b$ in max-gap PD-instances. The same principles apply equally well to bound $b$ in max-gap CD-instances.}


The following theorem strengthens Theorem~\ref{thm:structure} to account for the fact that now we are not only interested in a negative example, but rather in a negative example of a max-gap type. In this theorem, we identify $R$, $C$ and $U$ from the parallel diagonals structure with agents~1,~2 and~3, respectively. $b$ denotes the value of the MMS.

\begin{theorem}
\label{thm:max-gap}
Every max-gap PD-instance has the parallel diagonals structure (by definition). Moreover, all properties below either necessarily hold, or can be assumed to hold without loss of generality.

\begin{enumerate}
    \item For every agent $i$, for every bundle $B$ that is in her MMS partition according to the parallel diagonals structure, it holds that $v_i(B) = b$.
    \item For every allocation $A = (A_1, A_2, A_3)$, at least one of the three inequalities $v_i(A_i) \le b-1$ holds (where $1 \le i \le 3$).
    \item For every agent $i$ and every bundle $B$ that belongs to the MMS partition of a different agent and $B$ is bad for $i$, it holds that $v_i(B)\le b-1$.
    \item For every item $e$ and agent~$i$, it holds that $v_i(e) \ge 1$.
    \item {With only four possible exceptions, for every two items $e$ and $e'$ and every agent $i$, if the two items are not in the same bundle in the $MMS_i$ partition, then $|v_i(e) - v_i(e')| \ge 1$. There are two exceptions, $(e_6,e_9)$ and $(e_8,e_9)$, and two partial exceptions, the cases $v_R(e_8) \ge v_R(e_2)$ and $v_C(e_6) \ge v_C(e_4)$.}
    \item The instance is ordered. Namely, there is some permutation $\pi$ over nine items. For every two items $e$ and $e'$, if $e_i \ge_{\pi} e_j$ (this notation denotes that $e$ precedes $e'$ according to $\pi$), then for every agent $i$ it holds that $v_i(e) \ge v_i(e')$.
\end{enumerate}

\end{theorem}

\begin{proof}
We prove the properties of the theorem in the same order as they are stated in the theorem. For each property, the notation in the proof corresponds to the notation in the theorem.

\begin{enumerate}
    \item The fact that $b$ is the MMS implies $v_i(B) \ge b$. We may further assume without loss of generality that $v_i(B) = b$, because reducing the value of an item (while still keeping the MMS  value to be at least $b$), there still is no allocation that gives every agent a value larger than $b-1$.

    \item This is because in a max-gap PD-instance there is no allocation that gives every agent value strictly larger than $b-1$.

    \item Suppose for the sake of contradiction that $v_i(B) = b-1+\delta$ for $\delta > 0$ (and note that $\delta < 1$ by Proposition~\ref{pro:uniqueGood}). Let $e$ be an arbitrary item in $B$. Modify the valuation function $v_i$ to $v'_i$, where the only change is that $v'_i(e) = v_i(e) + 1 - \delta$. Hence now $v'_i(B) = b$, and Proposition~\ref{pro:uniqueGood} implies that an MMS allocation exists (with respect to $v'_i$). In this allocation, every agent other than $i$ gets at least her MMS, whereas agent $i$ gets at least $b - 1 + \delta > b-1$ (with respect to $v_i$). Hence the instance is not a max-gap instance.

    \item Suppose for the sake of contradiction that $v_i(e) = 1 - \delta$ for $\delta > 0$ (and note that $\delta \le 1$ as item values are non-negative). Let $e' \not= e$ be an arbitrary item such that both $e$ and $e'$ belong to the same bundle $B$ in the $MMS_i$ partition. Modify the valuation function $v_i$ to $v'_i$, where the only {changes are} that $v'_i(e') = v_i(e') + 1 - \delta$ {and $v'_i(e) = 0$}. Hence now $v'_i(B \setminus \{e\}) = b$. This allows us to move $e$ to a different bundle of the $MMS_i$ partition. After this move, we {still have an $MMS_i$ partition with respect to $v'$, but the structure is} neither the parallel diagonals structure nor the crossing diagonals structure (as sizes of bundles no longer obey these structures). Hence an MMS allocation exists. In this allocation, every agent other than $i$ gets at least her MMS, whereas agent $i$ gets at least $b - 1 + \delta > b-1$ (with respect to $v_i$). Hence the instance is not a max-gap instance.

    \item Let $e$ and $e'$ be two items different than the pairs $(e_6,e_9)$ and $(e_8,e_9)$. Suppose for the sake of contradiction that $v_i(e) = v_i(e') + 1 - \delta$ for $0 < \delta \le 1$, where $e$ and $e'$ are not in the same bundle in the $MMS_i$ partition.  Modify the valuation function $v_i$ to $v'_i$, where the only change is that the values of $e$ and $e'$ are swapped. Namely, $v'_i(e') = v_i(e)$ and $v'_i(e) = v_i(e')$. Now swapping $e$ and $e'$ in the bundles of the $MMS_i$ partition (recall that $e$ and $e'$ are in different bundles), we get a new $MMS_i$ partition, that we refer to as $MMS'_i$.  After this swap, we no longer have the structure required by Theorem~\ref{thm:structure}. This can be shown via a case analysis. We sketch the case analysis for the case that $i=R$, and leave the cases $i=C$ and $i=U$ to the reader.
    If $i = R$, then $e$ and $e'$ are in different rows. If they are also in different columns, then after the modification, {there is a bundle in $MMS'_R$ and a bundle in $MMS_C$ that do not intersect each other}, contradicting Theorem~\ref{thm:structure}. If they are in the same column, and either $e$ or $e'$ belong to $D$, then after the modification, a bundle in $MMS'_R$ fails to intersect $D$. Hence neither $e$ nor $e'$ are in $D$. Conditioned on this and on being in the same column, there are four possibilities. If $e$ and $e'$ are in $C_1$, then they must be $e_1$ and $e_4$. After the modification, $P$ is contained in a single bundle of the $MMS'_R$ partition. If they are in $C_2$ and $e_2$ is at least as valuable as $e_8$, then the bundle $(e_7,e_2,e_9)$ is in the $MMS'_R$ partition, is good for all three agents, and intersects $P$ (which is not allowed, by Theorem~\ref{thm:structure}). If they are in $C_2$ and $e_8$ is at least as valuable as $e_2$, then this is the partial exception concerning $v_R(e_8) \ge v_R(e_2)$. If they are in $C_3$ then they are $e_6$ and $e_9$, and this is the exception concerning $(e_6,e_9)$.


    In each of the cases analyzed above, the structure that results after replacing $MMS_R$ by $MMS'_R$ is not consistent with Theorem~\ref{thm:structure}.
    Hence an MMS allocation exists. In this allocation, every agent other than $R$ gets at least her MMS, whereas agent $R$ gets at least $b - 1 + \delta > b-1$ (with respect to $v_R$). Hence the instance is not a max-gap instance.

    \item Being an ordered instance can be assumed without loss of generality, by Proposition~\ref{pro:BL}.
\end{enumerate}

\end{proof}

Armed with Theorem~\ref{thm:max-gap}, we set up a linear program that searches for a max-gap PD-instance. For each agent $i$, there are nine variables specifying the values of $v_i(e_1), \ldots, v_i(e_9)$. For convenience, we rename these variables to be $\{r_1, \ldots, r_9\}$ for agent $R$, $\{c_1, \ldots, c_9\}$ for agent $C$, and $\{u_1, \ldots, u_9\}$ for agent $U$ (as in Figure~\ref{fig:9Elements3PartitionsB}). In addition, we have the variable $b$. Hence altogether there are 28 variables. The objective function of the LP is to minimize $b$.

Item~4 of Theorem~\ref{thm:max-gap} gives 27 linear constraints, that we refer to as the {\em positivity constraints}. That is, for every $1 \le j \le 9$ we have the constraints:

\begin{eqnarray*}\label{eqn:positivity}
\begin{array}{lcl}
     r_j &\ge  &1 \\
     c_j &\ge  &1 \\
     u_j &\ge  &1
 \end{array}
 \end{eqnarray*}

Item~1 of Theorem~\ref{thm:max-gap} gives 9 linear constraints, that we refer to as the {\em MMS constraints}.

\begin{eqnarray*}\label{eqn:MMSconstraints}
\begin{array}{lcl}
     r_1 + r_2 + r_3 &= &b \\
   r_4+r_5+r_6 &= &b \\
   r_7+r_8+r_9 &= &b \\
  c_1 + c_4 + c_7 &= &b \\
   c_2 + c_5+ c_8 &= &b \\
   c_3+c_6+c_9 &= &b \\
   u_2 + u_4 &= &b \\
   u_3+u_5+u_7 &= &b \\
   u_1+u_6+u_8+u_9 &= &b
 \end{array}
 \end{eqnarray*}

Item~3 of Theorem~\ref{thm:max-gap} gives 12 linear constraints, that we refer to as the {\em bad bundles constraints}.

\begin{eqnarray*}\label{eqn:badBundles}
\begin{array}{lcl}
     c_1 + c_2 + c_3 &\le &b-1 \\
          u_1 + u_2 + u_3 &\le &b-1 \\
   c_4+c_5+c_6 &\le &b-1 \\
    u_4+u_5+u_6 &\le &b-1 \\
  r_1 + r_4 + r_7 &\le &b-1 \\
  u_1 + u_4 + u_7 &\le &b-1 \\
   r_2 + r_5+ r_8 &\le &b-1 \\
    u_2 + u_5+ u_8 &\le &b-1 \\
   r_3+r_5+r_7 &\le &b-1 \\
    c_3+c_5+c_7 &\le &b-1 \\
   r_1+r_6+r_8+r_9 &\le &b-1 \\
    c_1+c_6+c_8+c_9 &\le &b-1
 \end{array}
 \end{eqnarray*}

Minimizing $b$ subject to the above $27 + 9 + 12 = 48$ constraints, the optimal solution (found by an LP solver) has value $b=13$. This certifies that in every max-gap PD-instance, every agent receives at least a $\frac{12}{13}$ fraction of her MMS. This is still far from our goal of $\frac{39}{40}$, that would match the bound in Theorem~\ref{thm:largeGap}. Indeed, the solution to the LP is not a max-gap PS-instance, and not even a negative example: it satisfies all constraints of the LP, but does have an MMS allocation.

To make further progress, we need to generate additional valid constraints to the LP. To be useful, these constraints should not be implied by the existing constraints. We present two procedures for generating useful constraints.

One procedure generates constraints that we refer to as {\em order constraints}.
Recall item~6 of Theorem~\ref{thm:max-gap}, which imposes an order $\pi$ over the items. The following proposition uses it in order to derive additional constraints.

\begin{proposition}
\label{pro:orderConstraints}
We may assume without loss of generality that $e_4 \ge_{\pi} e_2$, $e_3 \ge_{\pi} e_7$ and  $e_8 \ge_{\pi} e_6$.
\end{proposition}

\begin{proof}
By symmetry and item~6 of Theorem~\ref{thm:max-gap} (which imposes an order $\pi$ over the items), we may assume that $e_4 \ge_{\pi} e_2$. (If $e_2 \ge_{\pi} e_4$, then we may transpose all matrices and interchange the roles of $R$ and $C$, and get an equivalent negative instance that does satisfy $e_4 \ge_{\pi} e_2$.) By item~5 of Theorem~\ref{thm:max-gap}, this gives the following three constraints: $r_4 \ge r_2 + 1$, $c_4 \ge c_2 + 1$ and $u_4 \ge u_2$ (observe that $e_2$ and $e_4$ are in the same bundle in the $MMS_U$ partition).

The fact that $r_1 + r_2 + r_3 = v_R(R_1) = b > b-1 \ge v_R(C_1) = r_1 + r_4 + r_7$ then implies that $r_3 \ge r_7 + 1$.  As the instance can be assumed to be ordered, we have $e_3 \ge_{\pi} e_7$, which implies the two constraints $c_3 \ge c_7 + 1$ and $u_3 \ge u_7$  (also here, $e_3$ and $e_7$ are in the same bundle in the $MMS_U$ partition). In a similar manner, the fact that  $v_C(C_2) = b > b-1 \ge v_C(R_2)$ implies that $e_8 \ge_{\pi} e_6$, and consequently that $r_8 \ge r_6 + 1$ and $u_8 \ge u_6$.
\end{proof}

By Proposition~\ref{pro:orderConstraints} and its proof, we add {nine} constraints to the LP. {(Two of the constraints, $r_3 \ge r_7 +1$ and $c_8 \ge c_6 + 1$, are not really needed, as they are implied by constraints already present in the LP.)} We refer to these constraints as {\em order constraints}. Observe that some of the power comes from the fact that an order constraint such as $e_4 \ge_{\pi} e_2$ does not simply translate to $r_4 \ge r_2$, but rather to the stronger $r_4 \ge r_2 + 1$.

\begin{eqnarray*}\label{eqn:order}
\begin{array}{lcl}
     r_4 &\ge &r_2 + 1\\
          c_4 &\ge &c_2 + 1\\
               u_4 &\ge &u_2 \\
               r_3 &\ge &r_7 + 1 \\
                    c_3 &\ge &c_7 + 1 \\
                         u_3 &\ge &u_7 \\
                              r_8 &\ge &r_6 + 1\\
                              c_8 &\ge &c_6 + 1\\
                                   u_8 &\ge &u_6
 \end{array}
 \end{eqnarray*}

With the above order constraints, the value of the LP increases to $b = 16$.
There are plenty of other order constraints that can be inferred (see Appendix~\ref{app:implement}), but those that we could infer did not turn out to be useful.

Another procedure generates constraints that we refer to as {\em allocation constraints}. The following proposition introduces three such constraints.

\begin{proposition}
\label{pro:allocationConstraints}
The following three constraints are valid constraints.
\begin{enumerate}
    \item $u_6 + u_7 + u_9 \le b-1$.
    \item $u_3 + u_9 \le b-1$.
    \item $r_3 + r_9 \le b-1$.
\end{enumerate}
\end{proposition}

\begin{proof}
Consider the allocation $(e_1,e_3,e_4), (e_2,e_5,e_8), (e_6,e_7,e_9)$ to agents $R$, $C$ and $U$ respectively. Observe that $v_R(e_1,e_3,e_4) = r_1 + r_3 + r_4 > r_1 + r_3 + r_2 = b$ and that $v_C(e_2,e_5,e_8) = c_2 + c_5 + c_8 = b$. Hence to be a max-gap PD-instance, $v_U(e_6,e_7,e_9) \le b-1$ must hold, implying the first constraint of the proposition.

Consider the allocation $(e_4,e_5,e_6), (e_1,e_2,e_7,e_8), (e_3,e_9)$.  Observe that $v_R(e_4,e_5,e_6) = r_4 + r_5 + r_6 = b$. We also have $v_C(e_1,e_2,e_7,e_8) =  c_1 + c_2 + c_7 + c_8 > b$. This holds because $c_1 + c_7 > c_5 + c_6$ (because $v_c(C_1) = b > b-1 = v_C(R_2)$), and $c_2 + c_5 + c_8 = b$. Hence to be a max-gap PD-instance, $v_U(e_3,e_9) \le b-1$ must hold, implying the second constraint of the proposition.

We now prove the third constraint of the proposition. Assume for the sake of contradiction that $r_3 + r_9 > b-1$. Having proved that $v_U(e_3,e_9) \le b-1$, we have that $u_1 + u_2 + u_4 + u_5 + u_6 + u_7 + u_8 > 2b$. Hence either $u_1 + u_4 +  u_7 > b$  or $u_2 + u_5 + u_6 + u_8 > b$. Whichever inequality holds (pick the first one if they both hold), give the respective items to $U$ (who gets more than her MMS), items $(e_3,e_9)$ to $R$ (who gets more than $b-1$), and the remaining items to $C$ (who gets at least her MMS, as she gets a complete column). This allocation has a gap strictly smaller than~1, contradicting the assumption that the instance is a max-gap PD-instance.
\end{proof}

With the three allocation constraints of Proposition~\ref{pro:allocationConstraints}, the value of the LP increases to $b = 19$.

We refer to the above LP {(containing $60$ constraints)} as the {\em root LP}. We are not aware of additional useful linear constraints that hold without loss of generality.
{Hence at this point, we revert to mixed integer programming.

\subsection{Using mixed integer programming}

We use the root LP described above as the basis of a mixed integer program (MIP) that finds the negative example with largest MMS gap, among those that have the parallel diagonals structure. (A similar MIP handles the crossing diagonals structure.) This is done by introducing {\em selection variables} that take on only $0/1$ values. Using these selection variables, we add two sets of constraints, referred to as {\em selected order constraints} and {\em selected allocation constraints}.

The goal of the selected order constraints is to make sure that the MIP selects a consistent ordering $\pi$ among the items. As we do not know which ordering to select, the selection variables allow the MIP to investigate all possible orderings (that are consistent with the order constraints that we already have), and choose among them the worst one. To illustrate how this is done, consider the two items $e_4$ and $e_9$. We wish to allow the MIP to investigate both the possibility that $e_4 \ge_{\pi} e_9$ and that $e_9 \ge_{\pi} e_4$. We introduce a selection variable $s_{49}$ with the integer (binary) constraint $s_{49} \in \{0,1\}$. We also select a sufficiently large constant, larger than the maximum possible difference in value between any two items. Choosing this constant as~40 suffices for our purpose. The MIP then contains the following six linear constraints (the $+1$ terms are from item~5 of Theorem~\ref{thm:max-gap}):

\begin{itemize}
    \item  $r_4 \ge r_9 + 1 - 40 + 40s_{49}$.
    \item  $c_4 \ge c_9 + 1 - 40 + 40s_{49}$.
    \item  $u_4 \ge u_9 + 1 - 40 + 40s_{49}$.
    \item  $r_9 \ge r_4 + 1 - 40s_{49}$.
    \item  $c_9 \ge c_4 + 1 - 40s_{49}$.
    \item  $u_9 \ge u_4 + 1 - 40s_{49}$.
\end{itemize}

If $s_{49} = 1$ the above constraints are satisfied if and only if $e_4 \ge_{\pi} e_9$. If $s_{49} = 0$, they are satisfied if and only if $e_9 \ge_{\pi} e_4$.

We do not need to add selected order constraints for every pair of items, because some order constraints are forced by other constraints (e.g., by Proposition~\ref{pro:orderConstraints}). In our MIP, we used selected order constraints for seven
pairs of items.

The goal of the selected allocation constraints is to make sure that solutions of the MIP are supported only on max gap instances. To illustrate how this is done, consider a candidate allocation $(e_1,e_2,e_3), (e_4,e_6,e_8), (e_5,e_7,e_9)$ (the bundles go to agents $R$, $C$ and $U$ in this order).
We need to ensure that in this allocation some agent gets value at most $b-1$. Observe that in this allocation agent $R$ gets $R_1$ and hence a value of at least $b$. Thus the agent getting value at most $b-1$ is either $C$ or $U$, but we do not know which one of them. We introduce a selection variable $s_{c468u579}$ with the integer (binary) constraint $s_{c468u579} \in \{0,1\}$. We also select a sufficiently large constant, larger than the maximum possible value of a bundle of items. Choosing this constant as~100 suffices for our purpose. The MIP then contains the following two linear constraints.

\begin{itemize}
    \item $c_4 + c_6 + c_8 \le b-1 + 100s_{c468u579}$.
    \item $u_5 +u_7 + u_9 \le b-1 + 100 - 100s_{c468u579}$.
\end{itemize}

If $s_{c468u579} = 1$ the above constraints are satisfied if and only if the bundle of $U$ is worth at most $b-1$. If $s_{c468u579} = 0$, they are satisfied if and only of the bundle of $C$ is worth at most $b-1$.

As in the case of selected order constraints, we do not need to add selection allocation constraints for every possible allocation, as some allocation constraints are forced by other constraints (e.g., by Proposition~\ref{pro:allocationConstraints}). In our MIP, we used selected allocation constraints for forty allocations.

The code for the MIPs (one for the parallel diagonals structure, one for the crossing diagonals structure) that verify the correctness of Theorem~\ref{thm:tightGap} can be obtained from the authors upon request. Appendix~\ref{app:implement} provides some further explanations regarding the constraints used in these MIPs. 

\begin{remark}
The example proving Theorem~\ref{thm:largeGap} was also found using the MIP approach described above. We extracted a negative example from a solution of value $b=40$, and then manually modified it so as to make it easier to generate a humanly verifiable proof for its correctness. (There are several negative examples with $b=40$, and the MIP solver does not necessarily produce the one that is easiest for humans to analyse.)
\end{remark}}

\section{Extension to chores}
\label{sec:chores}

{\em Chores} are items of negative value, or equivalently, positive {dis-utility}. In allocation problems involving only chores, the convention is that all items must be allocated. In analogy to Definition~\ref{def:MMS},  $MMS_i$ is the minimum over all $n$-partitions of $M$, of the maximum dis-utility under $v_i$ of a bundle in the $n$-partition. (Note that as dis-utility replaces value, maximum and minimum are interchanged in this definition, compared to Definition~\ref{def:MMS}.) It is known that for agents with additive dis-utility functions over chores, there are allocation instances in which in every allocation some agent gets a bundle of dis-utility higher than her MMS~\cite{ARSW17}, and that there always is an allocation giving every agent a bundle of dis-utility at most $\frac{11}{9}$ times her MMS~\cite{HL19}.

{We now prove Theorem~\ref{thm:chores}, that there is an instance with three agents and nine chores that has an MMS gap of $\frac{1}{43}$.}


\begin{proof}
We present an example with an MMS gap of $\frac{1}{43}$, using notation as in Section~\ref{sec:40}.

Every row in the dis-utility function of $R$ has value~43 and gives $R$ her MMS. Her dis-utility function is:

$\left(
  \begin{array}{ccc}
    6 & 15 & 22 \\
    26 & 10 & 7 \\
    {\bf 12} & {\bf 19} & {\bf 12} \\
  \end{array}
\right)$

Every column in the dis-utility function of $C$ has value~43 and gives $C$ her MMS. Her dis-utility function is:

$\left(
  \begin{array}{ccc}
    6 & 15 & {\bf 23} \\
    26 & 10 & {\bf 8} \\
    11 & 18 & {\bf 12} \\
  \end{array}
\right)$

The bundles that give $U$ her MMS are $p = \{e_2, e_4\}$ (the {\em pair}, in boldface), $d = \{e_3, e_5, e_7\}$ (the {\em diagonal}), and $q = \{e_1, e_6, e_8, e_9\}$ (the {\em quadruple}). The dis-utility function of $U$ is:

$\left(
  \begin{array}{ccc}
    6 & {\bf 16} & 22 \\
    {\bf 27} & 10 & 7 \\
    11 & 18 & 12 \\
  \end{array}
\right)$

In analogy to Section~\ref{sec:40}, to analyse all possible allocations in a systematic way, it is convenient to consider a dis-utility function $M$ in which the dis-utility of each chore as the minimum (rather than maximum, as we are dealing with chores) dis-utility given to the chore by the three agents. Hence $M$ is:

$\left(
  \begin{array}{ccc}
    6 & 15 & 22 \\
    26 & 10 & 7 \\
    11 & 18 & 12 \\
  \end{array}
\right)$

Adaptation of the analysis of Section~\ref{sec:40} shows that in every allocation, some agent gets chores of dis-utility at least~44, whereas the MMS is~43. Further details of the proof are omitted.
\end{proof}

\section{Discussion}

{The open questions below refer to allocations of goods. Questions of a similar nature can be asked for chores, though the quantitative bounds in these questions would be different from those mentioned below.}

Let $\delta_n$ denote the largest value such that for $n$ agents, there is an allocation instance with additive valuations for which no allocation gives every agent more than a $1 - \delta_n$ fraction of her MMS. We have that $\delta_1 = \delta_2 = 0$. As to $\delta_3$, the combination of our Theorem~\ref{thm:largeGap} and the results of~\cite{GM19} imply that $\frac{1}{40} \le \delta_3 \le \frac{1}{9}$. It would be interesting to determine the exact value of $\delta_3$, or at least to narrow the gap between its lower bound and upper bound. Computer assisted techniques, such as those used in the proof of Theorem~\ref{thm:tightGap}, may turn out useful for this purpose.

For general $n \ge 3$, the combination of our Theorem~\ref{thm:polynomialGap} and the results of~\cite{GT20} imply that $\frac{1}{n^4} \le \delta_n \le \frac{1}{4} + \frac{1}{12n}$. We do not know whether $\delta_n$ tends to~0 as $n$ grows. Determining whether this is the case remains as an interesting open question. The known results do not exclude the possibility that $\delta_n$ tends to $\frac{1}{4}$ as $n$ grows, but we would be very surprised if this turns out to be true.

\subsection*{Acknowledgements}

{This research benefited from the use of automated solvers for mixed integer programs. Specifically, we have used solvers of {\em https://online-optimizer.appspot.com/} and  {\em https://www.lindo.com/}.}

{We would like to thank Shai Keidar and Orin Munk for helpful discussions.}

\begin{appendix}

\section{Some further details on the MIP}
\label{app:implement}

{
The notation in this appendix is hopefully self explanatory. Proofs are only sketched.

We note that for the parallel diagonal structure, many order constraints are implied by other constraints. This may explain why a relatively small number of selected order constraints suffice in order to enforce a consistent total order among all items.
The following order constraints are stated here for convenience, with hints as to why they hold. They need not be added explicitly to the MIP, as they are implied by other constraints in the root LP.

\begin{itemize}

\item $e_2 \ge e_5 + 1$ ($P > R_2$).

\item $e_2 \ge e_7 + 1$  ($P > C_1$).

\item $e_3  \ge e_8 + 1$ ($C_3 > Q$).

\item $e_4 \ge e_3 + 1$ ($P > R_1$).

\item $e_7 \ge e_1 + 1$ ($R_3 > Q$).

\item $e_7 \ge e_6 + 1$ ($R_3 > Q$).

\item $e_9 \ge e_5 + 1$ ($R_3 > D$ and $e_3 \ge e_8 + 1$).

\item $e_5 \ge e_1 + 1$ ($v_R(R_2) > v_R(C_1)$ and $e_7 > e_6$; $v_C(C_2) > v_C(R_1)$ and $e_3 > e_8$; $v_U(D) > v_U(C_1)$ and $e_4 > e_3$).

\end{itemize}

We now provide some details (proofs left to the reader) about the root LP for the crossing diagonals structure. Naturally, it has the MMS constraints and the bad bundles constraints that are associated with the CD structure (e.g., the constraint $u_1 + u_5 + u_9 = b$). Theorem~\ref{thm:max-gap} holds with the following change: the two exceptions are $(e_3, e_6)$ and $(e_7,e_8)$, instead of $(e_6,e_9)$ and $(e_8,e_9)$. Proposition~\ref{pro:orderConstraints} and~\ref{pro:allocationConstraints} hold without change. Consequently, we can have a root LP with 60 constraints (same number as in the root LP for the parallel diagonals structure), and then extend it to an MIP. For the crossing diagonals structure, the value of the MIP turns out to be~47, larger than the value of~40 obtained for the parallel diagonals structure.

For the crossing diagonals structure we can strengthen the root LP, as additional useful order constraints can be inferred. The following order constraints hold (hints are given) but need not be added (as they are implied by other constraints).

\begin{itemize}
    \item $e_2 > e_1, e_7$ ($P \ge C_1$).
    \item $e_2 > e_5, e_6$ ($P \ge R_2$).
    \item $e_4 > e_1, e_3$ ($P \ge R_1$).
    \item $e_4 > e_5, e_8$ ($P \ge C_2$).
    \item $e_9 > e_3, e_6$ ($R_3 \ge Q$).
    \item $e_9 > e_7, e_8$ ($C_3 \ge Q$).
\end{itemize}

We now derive useful order constraints. Suppose for the sake of contradiction that $e_9 \ge e_4$. Recall that $e_3 \ge e_7$. Then $v_C(C_1) = v_C(C_3)$ implies that $e_1 \ge e_6$, whereas  $v_R(R_2) > v_R(D)$ implies that $e_6 > e_1$, a contradiction. Hence we have the following useful constraint:

\begin{itemize}
    \item $e_4 \ge e_9 + 1$.
\end{itemize}

This implies additional order constraints (note that most of them can be enhanced by a $+1$ term), some of which are useful:

\begin{itemize}
    \item $e_1 \ge e_6$ ($v_U(D) \ge v_U(R_2)$).
    \item $e_5 \ge e_7$ ($v_R(R_2) \ge v_R(C_1)$).
    \item $e_3 \ge e_5$ ($v_C(C_3) \ge v_C(R_2)$).
    \item $e_8 \ge e_1$ ($v_R(R_3) \ge v_R(C_1)$).
    \item $e_9 \ge e_2$ ($v_U(D) \ge v_U(C_2)$).
\end{itemize}

In general, there is a tradeoff in the effort involved in deriving constraints analytically, compared to having more binary selection variables in the MIP. We believe that we have reached a reasonable point along this tradeoff curve, so that neither the analytic proofs nor the code for the MIP are too complicated.

}

\end{appendix}


\begin{thebibliography}{22}


\bibitem{AMNS17} Georgios Amanatidis, Evangelos Markakis, Afshin Nikzad, Amin Saberi:
	Approximation Algorithms for Computing Maximin Share Allocations. ACM Trans. Algorithms 13(4): 52:1--52:28 (2017).







\bibitem{ARSW17} Haris Aziz, Gerhard Rauchecker, Guido Schryen, Toby Walsh:
Algorithms for Max-Min Share Fair Allocation of Indivisible Chores. AAAI 2017: 335--341.


\bibitem{BK20}  Siddharth Barman, Sanath Kumar Krishnamurthy:
Approximation Algorithms for Maximin Fair Division. ACM Trans. Economics and Comput. 8(1): 5:1--5:28 (2020).




\bibitem{BL15} Sylvain Bouveret, Michel Lemaitre: Characterizing conflicts in fair division of indivisible goods using a scale of criteria. Autonomous Agents and Multi-Agent Systems. 30 (2): 259, (2015).


\bibitem{Budish11} Eric Budish:
The combinatorial assignment problem: Approximate competitive equilibrium from equal incomes.  Journal of Political Economy 119 (6), 1061--1103 (2011). 











\bibitem{GT20} Jugal Garg, Setareh Taki:
An Improved Approximation Algorithm for Maximin Shares. EC 2020: 379--380.

\bibitem{GM19} Laurent Gourves, Jerome Monnot:
On maximin share allocations in matroids. Theor. Comput. Sci. 754: 50--64 (2019).

	
\bibitem{GHSSY18} Mohammad Ghodsi, Mohammad Taghi Hajiaghayi, Masoud Seddighin, Saeed Seddighin, Hadi Yami: Fair Allocation of Indivisible Goods: Improvements and Generalizations. EC 2018: 539--556.

\bibitem{HL19} Xin Huang, Pinyan Lu:
An algorithmic framework for approximating maximin share allocation of chores. CoRR abs/1907.04505 (2019).


\bibitem{KPW18} David Kurokawa, Ariel D. Procaccia, Junxing Wang:
	Fair Enough: Guaranteeing Approximate Maximin Shares. J. ACM 65(2): 8:1--8:27 (2018).









\end{thebibliography}
\end{document}